\documentclass[a4paper]{article}

\usepackage{fullpage}
\usepackage{graphicx}

\usepackage{microtype}

\usepackage{amssymb}
\usepackage{mathtools}
\usepackage{amsthm}

\usepackage{xspace}
\usepackage{color}

\usepackage{todonotes}

\newcommand{\cparagraph}[1]{\paragraph*{#1}}

\newtheorem{theorem}{Theorem}

\newtheorem{lemma}[theorem]{Lemma}

\newtheorem{observation}[theorem]{Observation}

\usepackage{hyperref}
\hypersetup{%
  breaklinks,%
  colorlinks=true,%
  linkcolor=[rgb]{0.45,0.0,0.0},%
  citecolor=[rgb]{0,0,0.45}
}

\makeatletter
\partopsep\z@ \textfloatsep 10pt plus 1pt minus 4pt
\def\section{\@startsection {section}{1}{\z@}%
  {-3.5ex plus -1ex
    minus -.2ex}{2.3ex plus .2ex}{\large\bf}}
\def\subsection{\@startsection{subsection}{2}%
  {\z@}{-3.25ex plus
    -1ex minus -.2ex}{1.5ex plus .2ex}{\normalsize\bf}}
\def\@fnsymbol#1{\ensuremath{\ifcase#1\or 1\or 2\or 3\or 4\or
    5\or 6\or 7 \or 8 \or 9 \or 10\or 11 \else\@ctrerr\fi}}
\makeatother

\title{Compatible Triangulations of Simple Polygons}

\author{Peyman Afshani%
  \thanks{Department of Computer Science, Aarhus University, Aarhus, Denmark. peyman@cs.au.dk}
  \and
  Boris Aronov%
  \thanks{Department of Computer Science and Engineering, Tandon
    School of Engineering, New York University, Brooklyn, NY,
    USA. boris.aronov@nyu.edu} 
  \and
  Kevin Buchin%
  \thanks{Faculty of Computer Science, TU Dortmund, Germany. kevin.buchin@tu-dortmund.de}
  \and
  Maike Buchin%
  \thanks{Faculty of Computer Science, Ruhr University Bochum, Germany}
  \and
  Otfried Cheong%
  \thanks{SCALGO, Aarhus, Denmark. otfried@scalgo.com}
  \and
  Katharina Klost%
  \thanks{Institut für Informatik, Freie Universität
    Berlin, and Universität Tübingen, Germany. kathklost@inf.fu-berlin.de}
  \and
  Carolin Rehs%
  \thanks{Faculty of Computer Science, TU Dortmund, Germany,
    and Department of Mathematics and Computer Science, TU Eindhoven,
    The Netherlands. carolin.rehs@tu-dortmund.de}
  \and
  Günter Rote%
  \thanks{Institut für Informatik, Freie Universität
    Berlin. rote@inf.fu-berlin.de}}

\newcommand{\TT}{\mathcal{T}}

\newcommand{\ComputeBlock}{\textit{ComputeBlock}}
\newcommand{\update}{\textit{update}}
\newcommand{\blockupdate}{\textit{block-update}}
\newcommand{\bigO}{O}

\let\leq\leqslant
\let\geq\geqslant
\let\le\leqslant
\let\ge\geqslant
\let\theta\vartheta

\begin{document}

\maketitle

\begin{abstract}
  Let $P$ and $Q$ be simple polygons with $n$~vertices each. We wish
  to compute triangulations of $P$ and~$Q$ that are combinatorially
  equivalent, if they exist.  We consider two versions of the problem:
  if a triangulation of~$P$ is given, we can decide in~$\bigO(n\log n
  + nr)$ time if~$Q$ has a compatible triangulation, where~$r$ is the
  number of reflex vertices of~$Q$.  If we are already given the
  correspondence between vertices of~$P$ and~$Q$ (but no
  triangulation), we can find compatible triangulations of~$P$ and~$Q$
  in time~$\bigO(M(n))$, where $M(n)$ is the running time for
  multiplying two~$n\times n$ matrices.
\end{abstract}

\section{Introduction}

We consider two simple polygons~$P$ and~$Q$ that both have
$n$~vertices, and want to compute triangulations of~$P$ and~$Q$ that
are \emph{compatible}:  we can number the vertices of both~$P$ and~$Q$
counter-clockwise from~$0$ to~$n-1$, and a diagonal~$(i,j)$ is in the
triangulation of~$P$ if and only if it is in the triangulation of~$Q$, 
see Figure~\ref{fig:example}.

\begin{figure}
 \centerline{\includegraphics{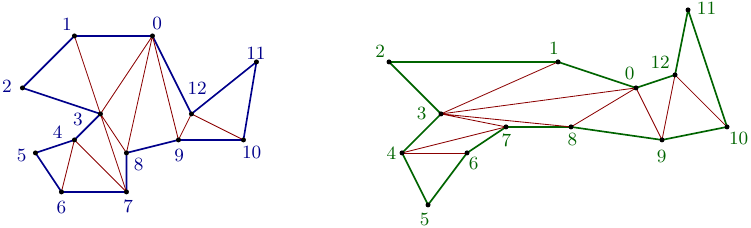}}
 \caption{Compatible triangulations of two polygons}
 \label{fig:example}
\end{figure}

We study two variants of this problem:
In the first variant, we are given a triangulation of~$P$, and ask if
there exists a numbering of the vertices of~$Q$ such that it has a
triangulation compatible to the triangulation of~$P$.  Clearly this
problem can be solved by trying all the $n$~distinct numberings
of~$Q$.  For a fixed numbering, we then only need to test if the
diagonals appearing in~$P$ exist in~$Q$. This takes~$\bigO(n^2)$ time
in total.

We improve this by giving an~$\bigO(n\log n + nr)$ time algorithm,
where~$r$ is the number of reflex vertices of~$Q$.  As part of our
algorithm, we give a data structure that stores a simple polygon~$P$
with~$n$ vertices in~$\bigO(n \log n)$ space, takes~$\bigO(n \log n)$
preprocessing time, and can answer in constant time questions of the form:
are vertices~$i$ and~$j$ visible in~$P$?

In the second variant, the vertex numbering of~$P$ and~$Q$ is given
and fixed, and we ask whether there are compatible triangulations
of~$P$ and~$Q$. It is easy to transform this question to the
following: Given a graph~$G$ drawn in the plane with straight edges
(possibly with crossings) whose vertices form a convex $n$-gon~$C$,
does it contain a triangulation of~$C$?

It is not difficult to solve this question in time~$\bigO(n^{3})$
using dynamic programming~\cite{ARONOV199327}.  We improve this
to~$\bigO(n^{\omega})$, where~$\omega$ is the exponent of matrix
multiplication.

\cparagraph{Related Work.}

It is well-known that it is not always possible to find compatible
triangulations for two simple polygons~$P$ and~$Q$ of~$n$
vertices. However, Aronov, Seidel, and Souvaine show that by adding
$\bigO(n^2)$ Steiner points, it is always possible to find compatible
triangulations~\cite{ARONOV199327}. They further show that
sometimes~$\Omega(n^2)$ Steiner points are necessary.  They also pose
the question how to find the smallest number of Steiner points
required for a given pair of simple polygons. This question is still
open for simple polygons, whereas for polygons with holes it has been
proven to be NP-hard by Lubiw and Mondal~\cite{LUBIW202097}.

\section{A data structure for vertex visibility queries in a simple polygon}

\begin{theorem}
  \label{thm:ds} A simple $n$-gon~$P$ can be preprocessed in
  time~$\bigO(n \log n)$ into a data structure with~$\bigO(n \log n)$
  space such that we can answer in constant time queries of the form:
  are vertex~$i$ and~$j$ visible in~$P$?
\end{theorem}
\begin{proof}
  We start by computing a triangulation of~$P$. We then recursively split the
  triangulation using diagonals, obtaining a binary
  tree~$\TT$ of height $\bigO(\log n)$.  Each node of~$\TT$
  corresponds to a region~$R$ of~$P$, which is then split by a
  diagonal~$d=ab$ into the two regions corresponding to the children
  of the node.  We compute, for each vertex~$v$ in the region~$R$
  different from $a$ and~$b$, the interval~$d_v$ on~$d$ that is
  visible from~$v$, and store this information with~$R$, see
  Figure~\ref{fig:splitting}.
  \begin{figure}[htb]
    \centerline{\includegraphics{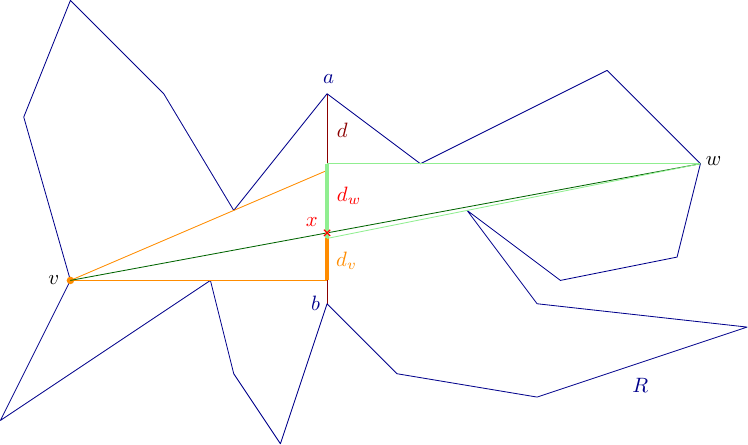}}
    \caption{Splitting the region~$R$ by a diagonal~$d$}
    \label{fig:splitting}
  \end{figure}

    To answer a visibility query for two vertices~$v$ and~$w$, we find
    the smallest region~$R$ that contains both~$v$ and~$w$.  Let $d$
    be the diagonal that splits~$R$.  If $v$ and $w$ are disjoint
    from~$d$, the segment~$vw$ lies in~$P$ if and only if it
    intersects the diagonal~$d$, and the intersection point~$x$ lies
    in both intervals~$d_{v}$ and~$d_{w}$, see
    Figure~\ref{fig:splitting}.  This can be checked in constant time.
    
    If one of the points, say $w$, is an endpoint of $d$, the
    points~$v$ and $w$ are visible if $w$ lies in~$d_v$.

    To find the region~$R$ for answering the query, we associate each
    vertex~$v$ of~$P$ to the highest node whose region is split by a
    diagonal incident to~$v$, or, if there is no such diagonal, to the
    unique leaf triangle containing~$v$.  The correct region for two
    query vertices~$v$ and~$w$ is then the lowest common ancestor of
    the two associated nodes in~$\TT$, which can be found in constant
    time after $O(n \log n)$ preprocessing time~\cite{BENDER200575}.

    Let us finally discuss how to store the intervals $d(v)$: If a
    vertex $v$ is associated to a node~$R$ at depth~$i$, let the
    sequence of regions on the path to that node be
    $R^0=P,R^1,R^2,\ldots,R^{i-1},R^i=R$, and let
    $d^0,d^1,d^2,\ldots,d^{i-1}$ be the corresponding splitting
    diagonals.  Then we store the intervals $d^0(v),d^1(v),
    \ldots,d^{i-1}(v)$ in an array of length~$\bigO(\log n)$.  In this
    way, with $\bigO(n\log n)$ space, we can access the required
    intervals in constant time.
\end{proof}

\section{One triangulation is given}

\begin{theorem}
  \label{thm:onefixed}
  Given two simple~$n$-gons~$P$ and~$Q$, and a triangulation $T$
  of~$P$, we can determine in \(O(n\log n + nr)\) time if there is a
  compatible triangulation in \(Q\), where $r$ is the number of reflex
  vertices of~$Q$.
\end{theorem}

\begin{proof}
  Let \(P=p_0,p_1,\dots p_{n-1}\) and \(Q=q_0,q_1,\dots, q_{n-1}\)
  where vertices are numbered counter-clockwise and vertex indices are
  considered modulo~$n$.  Our goal is to find a \emph{rotation} of the
  numbering of the vertices of~$Q$ by $s$~steps, such that for all
  diagonals~$p_ip_j$ in~$T$, the corresponding
  segment~$q_{i+s}q_{j+s}$ is a diagonal lying inside~$Q$.  In fact,
  we will identify \emph{all} such rotations~$s$.

  The following observation will be handy.
  \begin{observation}
    \label{obs:order}f
    For every vertex~$p_{i}$ of~$P$, there are indices
    $1<s_1<s_2<\cdots<s_\ell<n-1$ such that $p_ip_{i+s_1},
    p_ip_{i+s_2},\ldots p_ip_{i+s_\ell}$ are the diagonals incident to
    $p_i$ in the triangulation of $P$, in counterclockwise order
    around~$p_i$.  \qed
  \end{observation}

  Consider a reflex vertex~$q_{i}$ of~$Q$.  If we extend the two edges
  incident to~$q_{i}$ into the interior of~$Q$, they form a
  wedge~$w_{i}$.  Any triangulation of~$Q$ contains either
  \begin{itemize}
  \item a diagonal~$q_{i} q_{j}$ incident to~$q_{i}$ that lies
    inside the wedge~$w_{i}$, or
  \item
    a triangle $q_{i}q_{j}q_k$ whose opposite side $q_jq_{k}$ crosses
    the wedge~$w_{i}$.
  \end{itemize}
  
  Let us say that a rotation~$s$ \emph{satisfies} the reflex
  vertex~$q_{i}$ if the triangulation $T$ of~$P$, after applying the
  rotation~$s$, fulfills one of the two conditions.

  Conceptually, our algorithm fills a table~$G$ with $n$~rows
  corresponding to the different rotations~$s=0,\dots,n-1$, and
  $r$~columns corresponding to the reflex vertices of~$Q$.  We will
  \emph{mark} the cell~$G[s, i]$ if rotation~$s$ satisfies~$q_{i}$.
  In the end, we only need to identify the rows~$s$ of~$G$ where all
  cells are marked. Indeed, then there are diagonals of~$P$ that,
  after the rotation~$s$, partition~$Q$ into pieces that are all
  convex, since all interior angles are smaller than~$180^{\circ}$.
  It follows that these pieces can be triangulated in any way we want, in
  particular as prescribed by the given triangulation $T$ of~$P$.

  We start by computing the visibility query data structure of
  Theorem~\ref{thm:ds} for~$Q$.  We also precompute, for each reflex
  vertex~$q_{i}$ of~$Q$, the edges of~$Q$ hit by the extension of the
  edges~$q_{i-1}q_{i}$ and~$q_{i+1}q_{i}$, in $O(n)$ time per
  vertex~$q_i$. 

  By Observation~\ref{obs:order}, applied to $Q$ instead of~$P$, this
  gives us an interval of indices~$\{\sigma_i, \dots, \tau_{i}-1\}$
  such that a diagonal~$q_{i} q_{j}$ lies in the wedge~$w_{i}$ if and
  only if $\sigma_{i} \leq j < \tau_{i}$.  (If $\sigma_{i}= \tau_{i}$,
  then no such diagonals exist, and note that the segment~$q_{i}
  q_{j}$ is not necessarily a diagonal even if~$\sigma_{i} \leq j < \tau_{i}$.)

  Consider now a fixed rotation~$s$ and a reflex vertex~$q_{i}$
  of~$Q$. The rotation~$s$ maps the vertex~$p_{i-s}$ onto~$q_{i}$, so
  we consider the set of diagonals of $T$ incident to~$p_{i-s}$.
  \begin{itemize}
  \item If there is an index~$j$ such that $T$ contains a
    diagonal~$p_{i-s} p_{j-s}$ with~$\sigma_{i} \leq j < \tau_{i}$,
    then we test if the segment $q_{i}q_{j}$ lies in~$Q$. If this is
    not the case, we can immediately exclude the rotation~$s$ from
    further consideration. If it exists, we mark the cell~$G[s, i]$.
    Note that we do not need to test \emph{all} indices~$j$ in the
    range~$\sigma_{i} \leq j < \tau_{i}$.
  
  \item If no such index~$j$ exists, then we find the triangle
    $p_{i-s} p_{j-s}p_{k-s}$ in $T$ with $k < \sigma_{i} < \tau_{i}
    \leq j$. We test if the segments~$q_{i}q_{k}$, $q_{i}q_{j}$,
    and~$q_{k}q_{j}$ lie in~$Q$.  (They might be diagonals or they
    might simply be edges of~$Q$.)  We mark the cell~$G[s, i]$
    accordingly.
  \end{itemize}

  Clearly the procedure above can be implemented by binary search in
  the list of edges of $T$ incident to~$p_{i-s}$, taking~$\bigO(\log
  n)$ time for testing each pair~$(s, i)$, which leads to an~$\bigO(n
  r \log n)$ time algorithm. 

  To improve this to~$\bigO(nr)$, we
  \emph{batch} these tests.  We take the indices~$\sigma_{i}$
  and~$\tau_{i}$ over \emph{all} reflex vertices~$q_{i}$, and sort
  them into a single common list~$R$ of length~$2r$.  We now consider
  the vertices~$p_{i}$ of~$P$ in turn.  For each vertex~$p_{i}$, we
  merge the list of the $\ell_{i}$~diagonals incident to~$p_{i}$ with
  the list~$R$ in time~$\bigO(r + \ell_{i})$.

  We only need to walk through this list once to find the diagonals to
  be tested.  Since $\sum_{i}\ell_{i} = \bigO(n)$, this approach leads
  to the promised~$\bigO(n \log n +nr)$ algorithm.

  As a final optimization, we observe that we do not actually need to
  construct the table~$G$.  It suffices to keep a counter for each
  rotation, increasing its value whenever we mark a cell in the
  corresponding row.  In the end we report the rotations where the
  counter has reached the value~$r$.
\end{proof}

\section{Vertex correspondence is given}
\label{sec:give-correspondence}

We are given two simple~$n$-gons, whose vertices are already
numbered~$0$ to~$n-1$ in counter-clockwise order.

We first compute, in~$\bigO(n^{2})$ time, the visibility graphs of
both~$P$ and~$Q$~\cite{AsanoAGHI86}, obtaining the set of
diagonals~$(i, j)$ that exist in both~$P$ and~$Q$.  We have now
abstracted away from the polygons~$P$ and~$Q$ and solve our problem
using the following theorem.
\begin{theorem}
  \label{thm:graph}
  Given a graph~$G$ drawn in the plane
  with straight edges (possibly with crossings)
  whose vertices are the vertices of a convex $n$-gon~$C$.
  We can determine in $\bigO(n^{\omega})$ time
  whether~$G$ contains a triangulation of~$C$, and return such a
  triangulation if it exists.
\end{theorem}
\begin{proof}
  Let~$A$ be the adjacency matrix of~$G$, and let~$B$ be the Boolean 
  matrix where~$b_{ij} = 1$ if the
  edge~$(i,j)$ exists and there is a triangulation of the
  subpolygon bounded by this edge and the chain from~$i$ to~$j$.
  Our goal is to determine if~$b_{0,n-1} = 1$.

  We can express~$B$ recursively:
  \begin{align}
    \label{eq:recursion}
    b_{ij} &= \biggl( \,\bigvee_{i < k < j} (b_{ik} \wedge b_{kj}) \biggr)
    \wedge  a_{ij},\text { for $0\le i,j<n$, $j\ge i+1$}
    \\  \label{eq:recursion-diag}
    b_{i,i+1} &=   a_{i,i+1} (=1),\text { for $0\le i<n-1$}
  \end{align}
  Note that~\eqref{eq:recursion-diag} is actually just a special case of
  \eqref{eq:recursion} because the disjunction over the empty range $i <
  k < i+1$ evaluates to \emph{true}.  We set the undefined values
  $b_{ij}$ for $j\le i$ to~$0$. This results in a matrix $B=(b_{ij})$
  that is strictly upper-diagonal.

  Equations~\eqref{eq:recursion} and~\eqref{eq:recursion-diag} can be
  written in terms of Boolean matrix multiplication:
  \begin{equation} \label{matrix}
    B =( B \otimes B) \land A
  \end{equation}
  Since~\eqref{matrix} defines $B$ in terms of itself, $B$ cannot be
  \emph{computed} straightforwardly by matrix multiplication.
  Nevertheless~$B$ is well-defined because it is strictly
  upper-diagonal.  The usual (sequential) way to compute the
  recursions~\eqref{eq:recursion} is according to increasing length
  $j-i$, or rowwise from left to right, or columnwise from bottom to
  top. One has to ensure that the quantities on the right-hand side
  of~\eqref{eq:recursion} are computed before they are used.

  We will in fact work over the integers, and define the following
  \emph{triplet operation}:
  \begin{equation}
    \label{eq:triple-Z}
    \update(i,j,k)\colon\quad
    \hat b_{ij} := \hat b_{ij} + b_{ik} \cdot b_{kj}
  \end{equation}
  For a fixed pair~$(i, j)$, we initialize~$\hat b_{ij}$ to zero and
  perform the triplet operation for all~$k$ with~$i < k <j$.  Once all
  products have been accumulated, we convert $\hat b_{ij}$ to its
  final Boolean value $b_{ij}=a_{ij} \land \min\{\hat b_{ij},1\}$.

  The key to subcubic running time is to perform the triplet operation
  in \emph{blocks}.  We assume that~$n$ is a power of two, and
  recursively partition the range of indices~$0, \dots, n-1$ into
  blocks of decreasing size.  For fixed indices~$0 \leq i < k < j <
  n$, the triplet operation~$\update(i, j, k)$ is performed at the
  largest block size~$S$ such that the three indices~$i$, $k$, and~$j$
  lie in different blocks.  At this block size, we can perform the
  triplet operation for all indices within the same blocks using a
  single matrix multiplication:
  \begin{equation}
    \label{eq:triple-block}
    \blockupdate(S,u,v,w)\colon\quad
    \hat B^S_{uv} := \hat B^S_{uv} + B^S_{uw} \cdot B^S_{wv}
  \end{equation}
  Here, $S$~is the block size, which is always a power of two,
  and~$B^S_{uv}$ is the $S\times S$ submatrix of entries $b_{ij}$
  formed from the rows in the range $uS\le i < (u+1)S$ and the columns
  in the range $vS\le j < (v+1)S$.  
  
  The operation~$\blockupdate(S, u, v, w)$ is possible for $0\le
  u<w<v<n/S$ (for $w=u$ or $w=v$ it would define $\hat B^S_{uv}$ in
  terms of itself, or, more precisely, in terms of~$B^S_{uv}$).

  \begin{figure}
    \begin{tabbing}
      \qquad\=\+
      $\ComputeBlock(S,u,v)$:\\
      \qquad\=\+
      \textbf{if} \= $S=1$:\+ /\!/ Finalize $b_{uv}$:\\
      \textbf{if} $v=u+1$ \textbf{or}
      ($\hat b_{uv}>0 \land a_{uv}=1$):\\
      \qquad\=
      $ b_{uv}:=1$\\
      \textbf{else}\\
      \>$ b_{uv}:=0$\-\\
      \textbf{else if} $u<v$: \+/\!/ compute the four submatrices $\alpha, \dots, \delta$\\
      $\alpha$: \ \= $\ComputeBlock(S/2,2u+1,2v)$ \\[4pt]
      $\beta$:\>$\blockupdate(S/2,2u+1,2v+1,2v)$\\
      \>$\ComputeBlock(S/2,2u+1,2v+1)$ \\[4pt]
      $\gamma$:\>$\blockupdate(S/2,2u,2v,2u+1)$\\
      \>$\ComputeBlock(S/2,2u,2v)$ \\[4pt]
      $\delta$:\>$\blockupdate(S/2,2u,2v+1,2u+1)$\\
      \>$\blockupdate(S/2,2u,2v+1,2v)$ \\
      \>$\ComputeBlock(S/2,2u,2v+1)$ \-\\[2pt]
      \textbf{else} \+ /\!/ the diagonal case $u=v$ simplifies, because the
      submatrix $\alpha$ does not exist\\
      $\beta$:\>$\ComputeBlock(S/2,2u+1,2u+1)$ \\
      $\gamma$:\>$\ComputeBlock(S/2,2u,2u)$ \\
      $\delta$:\>$\ComputeBlock(S/2,2u,2u+1)$\-
    \end{tabbing}
    \caption{The recursive procedure \ComputeBlock}
    \label{program}
  \end{figure}

  Our recursive algorithm is shown in \autoref{program}. It consists
  of a single procedure~\ComputeBlock, which we start as
  $\ComputeBlock(n,0,0)$ after initializing the~$n \times n$
  matrix~$\hat{B}$ to zeroes. The notations~$\hat b_{ij}$
  and~$b_{ij}$ in the algorithm refer to the same storage location.
  The notation~$b_{ij}$ is used to remind us that the variable takes
  only the values~$0$ and~$1$, while~$\hat b_{ij}$ contains an
  arbitrary non-negative integer.

  \ComputeBlock~performs recursive calls on up to four submatrices
  labeled~$\alpha$, $\beta$, $\gamma$, and~$\delta$, see
  Figure~\ref{fig:blockmatrix}, using the recursive partitioning
  \begin{equation}
    \label{eq:partition4}
    B^S_{uv}=
    \begin{pmatrix}
      B^{S/2}_{2u,2v} &   B^{S/2}_{2u,2v+1}\\
      B^{S/2}_{2u+1,2v} &   B^{S/2}_{2u+1,2v+1}\\
    \end{pmatrix}=
    \begin{pmatrix}
      \gamma &   \delta\\
      \alpha &   \beta\\
    \end{pmatrix}
  \end{equation}

  \begin{figure}
    \centerline{\includegraphics{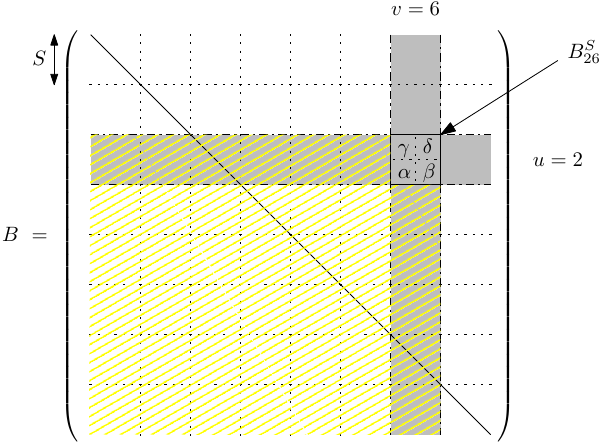}}
    \caption{The matrix $B$ is split into $8\times8$ square blocks of size $S=n/8$.
      The block $B^S_{uv}=B^S_{26}$ is shown to be composed of 4 subblocks of size
      $S/2$, labeled $\alpha,\beta,\gamma,\delta$ in the order in which they are recursively computed
      in the procedure $\ComputeBlock$.
    }
    \label{fig:blockmatrix}
  \end{figure}

  To prove correctness of the algorithm, we show that for each
  entry~$(i, j)$ of the matrix, we first compute the matrix product
  \[
  \hat{b}_{ij} = \sum_{i < k < j} b_{ik} \cdot b_{kj},
  \]
  over the integers.  In the base case of the recursion, where the
  block size~$S$ is one, we replace~$\hat{b}_{ij}$ by~$b_{ij} \in \{0,
  1\}$---this is where we perform the logical `and' operation with the
  adjacency matrix entry~$a_{ij}$.  At this point, we say that the entry~$(i, j)$ has then been
  \emph{finalized}. It remains unchanged from now on.

  \cparagraph{Precondition 1.} When $\ComputeBlock(S, u, v)$ is
  called, the entries $\hat{b}_{ij}$ in the block $B^S_{uv}$ fulfill
  the equation
  \begin{equation*}
  \hat b_{ij} = \sum_
       {
         \lfloor i/S\rfloor
         <\lfloor k/S\rfloor
         <\lfloor j/S\rfloor
       } b_{ik} \cdot b_{kj}.
  \end{equation*}
  In other words, all triplet operations $\update(i,j,k)$ for which
  $k$ lies in a block of size~$S$ different from~$i$ and~$j$ have
  already been performed.
  
  \cparagraph{Precondition 2.} When $\ComputeBlock(S, u, v)$ is
  called, all entries $b_{ij}$ with $i \geq uS$ and~$j < (v+1)S$ have
  already been finalized, except for the block~$B^S_{uv}$
  itself. (This is the hatched region in Figure~\ref{fig:blockmatrix}.)
       
  \cparagraph{Postcondition.} $\ComputeBlock(S, u, v)$ 
  finalizes the block~$B^S_{uv}$.  Together with Precondition~2, this
  means that all entries~$b_{ij}$ with~$i \geq uS$ and~$j < (v+1)S$
  have been finalized.

  \paragraph{Correctness.}
  Both preconditions clearly hold before the initial
  call~$\ComputeBlock(n, 0, 0)$.  The postcondition implies that when
  this call returns, the matrix~$B$ has been finalized and contains
  the result.

  Consider now a call to~$\ComputeBlock(S, u, v)$.

  If $S = 1$, precondition~1 guarantees that
  \[
  \hat{b}_{ij} = \sum_{i < k < j} b_{ik} \cdot b_{kj},
  \]
  so the call correctly computes~$b_{ij}$ and finalizes the entry.
  
  When~$S > 1$ and~$u = v$, then we observe that precondition~1 already
  holds for the recursive calls (in all recursive calls, $v \in \{u,
  u+1\}$, so there are no blocks of size~$S/2$ between~$u$ and~$v$).
  Precondition~2 for blocks~$\beta$ and~$\gamma$ already holds, and for
  block~$\delta$ if follows from the postconditions for~$\beta$
  and~$\gamma$.
  
  Finally, we consider~$S > 1$ and~$u < v$.  Observe that the
  \emph{only} blocks of size~$S/2$ between blocks~$2u$ and~$2v+1$ that
  are not already covered by blocks of size~$S$ are the blocks~$2u+1$
  and~$2v$.
  
  For block~$\alpha$, precondition~2 already holds, and there is no
  ``missing'' block of size~$S/2$.
  
  For block~$\beta$, precondition~2 follows from the postcondition of
  block~$\alpha$. The only missing block of size~$S/2$ is~$2v$, so we
  call $\blockupdate(S/2, 2u+1, 2v+1, 2v)$, which needs to
  access the blocks~$B^{S/2}_{2u+1,2v}$---this is block~$\alpha$, which
  has just been finalized---and block~$B^{S/2}_{2v,2v+1}$.  Since $2v
  \geq 2u + 2$, this block is already finalized by precondition~2.
  
  Similarly, for block~$\gamma$, the only missing block of size~$S/2$
  is~$2u+1$, so we call $\blockupdate(S/2,
  \allowbreak 2u, 2v, 2u+1)$.  This needs
  to access the blocks~$B^{S/2}_{2u,2u+1}$---since $2u+1 < 2v$, this
  block is already finalized by precondition~2---and
  block~$B_{2u+1,2v}^{S/2}$, which is block~$\alpha$.
  
  Finally, for block~$\delta$, precondition~2 follows from the
  postconditions of the three previous recursive calls.  The missing
  blocks are~$2u+1$ and~$2v$, so we call $\blockupdate(S/2, 2u, 2v+1,
  2u+1)$ and $\blockupdate(S/2, 2u, 2v, 2v)$.  These need to access
  blocks $B^{S/2}_{2u,2u+1}$, $\beta = B^{S/2}_{2u+1, 2v+1}$, $\gamma
  = B^{S/2}_{2u,2v}$, and $B^{S/2}_{2v, 2v+1}$.  Since $2u+1 < 2v$
  and~$2v \geq 2u+2$, all of these blocks are already finalized.

  \cparagraph{Runtime.}
  We let~$T(S)$ denote the running time of~$\ComputeBlock$ for a block
  of size~$S$.  Clearly $T(1) = \bigO(1)$ and $T(S)=\bigO(S^\omega)+4T(S/2)$
  for~$s > 1$, leading to $T(n)=\bigO(n^\omega)$, assuming $\omega>2$.
\end{proof}

\begin{theorem}
  Given simple \(n\)-gons \(P\) and \(Q\) with a fixed vertex
  numbering on both polygons, we can decide in \(O(n^\omega)\) time if
  there are compatible triangulations for \(P\) and \(Q\).
\end{theorem}

\paragraph{Counting triangulations.}

If we work with the integer entries $\hat b_{ij}$ as they are, only
setting them to 0 when $ a_{ij}=0$, but otherwise keeping their
values intact, the result will be the number of triangulations.
The numbers can be as large as nearly $4^n$; hence the bit-size of
the numbers must then be taken into account in the analysis.

\paragraph{The word problem for context-free grammars in Chomsky
  normal form.}

Our method is an alternative description of ``Valiant's trick'' for
the context-free grammar word problem~\cite{Valiant-1975}, which is
also known as ``interval dynamic programming.''  There is a slight
difference in the setup: we have to take the logical \emph{and}
with~$A$ before using an entry for further computations.  Valiant, on
the other hand, has a more general ``multiplication'' operation that
takes into account sets of nonterminal symbols.

Apart from these differences, the algorithms are structurally the
same, and our description is an alternative description of Valiant's
algorithm.  We use Valiant's procedures $P_2$, $P_3$, $P_4$ and the
lemma in~\cite[Section~4]{Valiant-1975}: the procedure $P_2$ in
Valiant's algorithm corresponds to our procedure~$\ComputeBlock$.
Indeed, when looking how Valiant's proof of his Theorem~2 reduces
$P_2$ for an $n\times n$ block to recursive calls of $P_2$ for $\frac
n2\times \frac n2$ blocks, via intermediate calls of $P_3$ and $P_3$,
one finds 4 recursive calls to $P_2$, just as in our procedure
$\ComputeBlock$. Each application of the lemma corresponds to one or
two calls of \textit{block-update} that precede the call of~$P_2$.

\paragraph{Can we do better?}

\begin{lemma}
  \label{lem:notbetter}
  The matrix \(B\) as defined above cannot be computed faster than
  Boolean matrix multiplication.
\end{lemma}

\begin{proof}
Given two Boolean $n \times n$ matrices~$M$ and~$N$, we can
construct a graph~$G$ by starting with a cycle of vertices
\[
s, x_1, x_2, ..., x_n, u, y_1, ... y_n, w, z_1, ... z_n, t.
\]
We then add diagonals as follows:
\begin{itemize}
\item diagonal $(x_i,y_k)$ for all $i, k$ where $M_{ik} = 1$,
\item diagonal $(y_k, z_j)$ for all $k, j$ where $N_{kj} = 1$,
\item all diagonals $(s,x_i)$, for all~$i$,
\item all diagonals $(s,z_j)$, for all~$j$,
\item all diagonals $(x_i,z_j)$, for all~$i$ and~$j$,
\item all diagonals $(x_i,u)$, for all~$i$,
\item all diagonals $(u,y_k)$, for all~$k$,
\item all diagonals $(y_k, w)$, for all~$k$,
\item all diagonals $(w, z_j)$, for all~$j$,
\item all diagonals $(z_j, t)$, for all~$j$.
\end{itemize}

If we compute the matrix~$B$ as defined above for this graph~$G$, then
for each pair~$(x_i, z_j)$ we will know whether there is a
triangulation of the subpolygon bounded by the chain~$x_i, x_{i+1},
\dots, \allowbreak z_{j-1}, z_{j}$.  If that is the case, then the triangle
incident to the edge $x_i z_j$ can only be $\triangle x_i y_k z_j$,
and the two edges necessary for that triangle exist only if $a_{ik} =
1$ and~$b_{kj} = 1$, so $(M \cdot N)_{ij} = 1$.  We leave it to the
reader to argue that whenever $(M \cdot N)_{ij} = 1$, then a
triangulation of the subpolygon exists.
\end{proof}

Put differently, any algorithm that returns not only a Yes/No-answer,
but that can tell us which other vertices can appear as the third
corner of the triangle incident to the edge~$st$, and also its two
neighboring triangles, must take at least the time required for
Boolean matrix multiplication.  Indeed, if~$G$ contains a
triangulation, the triangle incident to $st$ must be of the form
$\triangle s t z_j$.  The other triangle incident to the edge $sz_j$
must necessarily be of the form $\triangle s x_i z_j$, so we know
that~$(M \cdot N)_{ij} = 1$. On the other hand, if $(M \cdot N)_{ij} =
1$, then a triangulation of~$G$ exists that contains the
triangles~$\triangle s t z_{j}$ and~$\triangle s x_{i} z_{j}$.

This argument
implies that a faster algorithm can only be obtained by
exploiting the fact that the graph~$G$ is not an
arbitrary graph, but the intersection of the visibility graphs of
two simple polygons.

\bibliography{references}

\end{document}